\documentclass{birkjour}

\usepackage{ascmac}
\usepackage{mathrsfs}
\usepackage{latexsym}
\usepackage{amsmath}
\usepackage{amssymb}
\usepackage{amsthm}

 \newtheorem{thm}{Theorem}[section]
 \newtheorem{cor}[thm]{Corollary}
 \newtheorem{lem}[thm]{Lemma}
 
 \theoremstyle{definition}
 \newtheorem{defn}[thm]{Definition}
 \theoremstyle{remark}
 \newtheorem{rem}[thm]{Remark}
 
  \newtheorem{assm}[thm]{Assumption}
 \numberwithin{equation}{section}

\newcommand{\I}{\mathbb{I}}
\newcommand{\Z}{\mathbb{Z}}
\newcommand{\R}{\mathbb{R}}
\newcommand{\C}{\mathbb{C}}
\renewcommand{\I}{\mathbb{I}}

\newcommand{\N}{\mathbb{N}}

\newcommand{\Lam}{\Lambda_{N}(A)}

\begin{document}

%
%
%
%
%
%
%
%
%

\title[Exponential decay property for quantum walks]
 {Exponential decay property for eigenfunctions of quantum walks}

\author[Kazuyuki Wada]{Kazuyuki Wada}

\address{%
Department of Mathematics,
Hokkaido University of Education\\
9 cho-me, Hokumon-cho \\
Asahikawa, Hokkaido 070-8621, Japan
}

\email{wada.kazuyuki@a.hokkyodai.ac.jp}

\subjclass{Primary 81Q35; Secondary 47B02, 47B15, 47B93}

\keywords{Eigenfunction, Eigenvalue, Exponential decay, Quantum walk, Unitary operator}

\date{May 20, 2024}

\begin{abstract}
Under an abstract setting, we show that eigenvectors belong to discrete spectra of unitary operators have exponential decay properties. We apply the main theorem to multi-dimensional quantum walks and show that eigenfunctions belong to a discrete spectrum decay exponentially at infinity.
\end{abstract}

\maketitle
\section{Introduction}
Exponential decay property (EDP) at infinity is one of the characteristic properties of eigenfunctions associated with Schr\"{o}dinger operators. Earlier works on EDP are discussed by \v{S}nol'. In \cite{Sn}, he discussed the asymptotic behavior at infinity for eigenfunctions belong to discrete spectra. Moreover, it was clarified that there is a relation between the spectral gap and decay rate at infinity. O'Connor, Combes-Thomas, and Agmon considered EDP for $N-$body Schr\"{o}dinger operators. O'Connor showed EDP for pair potentials belonging to Rollnik class plus $L_{\epsilon}^{\infty}$ class \cite{O}. Combes and Thomas showed it for pair potentials which are analytic for the subgroup of linear transformation groups \cite{CT}. Agmon showed it by application of operator positivity methods \cite{A}. For other works on EDP, we refer Froese-Herbst \cite{FH}, Griesemer \cite{G1}, Nakamura \cite{N}, Bach-Matte \cite{BM}, Yafaev \cite{Y} and Kawamoto \cite{K}. We can also derive EDP from an application of the Feynman-Kac type formula. It is known that semigroups generated by a class of Schr\"{o}dinger operators can be represented by stochstic processes. In particular, martingale properties are crucial to deriving EDP. In this direction, we refer \cite{CMS, HIL, HL, LHB} and references therein. EDP also appears in the context of quantum field theory \cite{G2, H1, H2}. Besides, this property is not only shown but also applied to show the existence of ground states in non-relativistic quantum electrodynamics \cite{GLL, HHS}.

In this paper, we consider EDP for a class of unitary operators. Let $U$ be a unitary operator and $A$ be a non-negative self-adjoint operator on a Hilbert space $\mathcal{H}.$ We suppose that the discrete spectrum of $U$ is not empty. The purpose of this paper is to show
\begin{align}\label{eq:exp}
e^{\delta A}\psi\in\mathcal{H},
\end{align}
for any eigenvector $\psi$ belongs to the discrete spectrum and any sufficiently small $\delta>0$. In this case, we say that $\psi$ has EDP for $A$. As we see below, the range of $\delta$ is closely related to the distance between the essential spectrum of $U$ and the discrete eigenvalue which $\psi$ belongs to. A typical example of a non-negative self-adjoint operator $A$ in our mind is the modules of the position operator.

A motivation we consider EDP for unitary operators comes from quantum walks which are often regarded as a quantum counterpart of random walks \cite{Gu, LCETK, MSS}. From the viewpoint of partial differential equations, quantum walks are space-time discretized Dirac equations \cite{MS}. It is well known that some properties of quantum walks are quite different from that of random walks. In particular, the ballistic transportation and the localization occur in quantum walks \cite{K1, K2}. Related to these properties, mathematical analysis is developed from a viewpoint of weak limit theorem \cite{FFS3, S, RST2}, spectral theory\cite{MSSSS, MSW, RST1}, and references therein as examples. 

In the context of quantum walks, results on the existence of discrete spectra are known \cite{KS, M}. In particular, the explicit optimal decay rate is derived. In particular, in nonlinear quantum walks, EDP is applied to obtain the asymptotic stability \cite{M}. However, these references are limited in one dimension. In the one-dimensional case, we can introduce the transfer matrix which is a powerful tool for solving eigenvalue problems and analyzing various quantities. Although, in multi-dimensional cases, the existence of a discrete spectrum is reported in \cite{FFS2, FFSS}, detailed properties of eigenfunctions are not well known. In particular, it is not known whether eigenfunctions have EDP, yet. Motivated by these situations, we show EDP for a class of quantum walks involving multi-dimensional cases.

First, we establish \eqref{eq:exp} under a general setting in Section 2. Since we treat exponential operators of unbounded operators, we have to introduce suitable cut-off functions to avoid domain problems. For the proof, we mainly follow the methods presented by Yafaev \cite{Y} concerned the first-order differential systems involving Dirac operators. In our case, the derivative of functions are replaced by commutators. To analyze commutators is the crucial part.

In proofs, instead of $A$, we introduce another operator $\Lambda(A)$ which is step-like and approximates $A$ from above (see \eqref{Def:lambda}). In the function space, differential operators and multiplication operators act locally on configuration spaces. From this observation, in addition to introducing $\Lambda(A),$ it may be suitable to assume some locality conditions in $U$. Therefore, in this paper, we impose \lq\lq finite propagation" condition (see Assumption \ref{Ass:main}) for $U$. By these two ideas, we can analyze the commutator in detail.

The optimal constant $\delta$ in \eqref{eq:exp} depends on dispersion relations of quantum walks. For example, in \cite{KS, M}, the optimal constant is derived. However, in quantum walks, we can select graphs, internal degrees of freedom, motion of a quantum walker, and shift parameters. Thus, it would be useful to establish EDP in general settings. For example, in \cite{T}, Tiedra de Aldecoa considered spectral and scattering theory for quantum walks on not square lattices but trees. If discrete spectra of such quantum walks are not empty, we can apply our results. Our idea can be applied to discrete Schr\"{o}dinger operators since they consist of shift operators and multiplication operators that act locally.

As an application, in Section 3, we apply the results for multi-dimensional quantum walks with a defect. Then, we can show that eigenfunctions associated with discrete spectrum possess EDP.
\section{Set up and main result}
Let $\mathcal{H}$ be the separable Hilbert space over $\mathbb{C}$. The symbol $\langle\cdot, \cdot\rangle$ and $\|\cdot\|$ denotes the inner product and the norm over $\mathcal{H}$, respectively. Let $U$ be a unitary operator on $\mathcal{H}$. Symbols $\sigma(U)$, $\sigma_{\mathrm{ess}}(U)$ and $\sigma_{\mathrm{d}}(U)$ denote the spectrum of $U$, the essential spectrum of $U$ and the discrete spectrum of $U$, respectively. First, we introduce the following notion:
\begin{defn}\label{Def:propagate}
Let $S$ be a self-adjoint operator on $\mathcal{H}.$ We denote the spectral measure of $S$ by $E_{S}(\cdot).$ We say that $U$ finitely propagates with respect to $S$ if there exists a constant $b>0$ such that for any $\psi\in \mathrm{Ran}E_{S}([R_{1}, R_{2}))$ with $R_{1}< R_{2},$ $U\psi\in\mathrm{Ran}E_{S}([R_{1}-b, R_{2}+b)).$
\end{defn}
\begin{rem}
In Definition \ref{Def:propagate}, we introduced the notion of finite propagation for half-open intervals. Of course, we can also define the notion of the finite propagation by open intervals and closed intervals. However, we only consider half-open intervals to cover $[0, \infty)$ by disjoint intervals.
\end{rem}
 We impose the following assumption:
\begin{assm}\label{Ass:main}
\begin{enumerate}
\item $\sigma_{\mathrm{d}}(U)\neq \emptyset$.
\item The unitary operator $U$ finitely propagates with a constant $b>0$ with respect to a non-negative, possibly unbounded, self-adjoint operator $A.$
\end{enumerate}
\end{assm}
For any $\lambda\in\sigma_{\mathrm{d}}(U),$ we define the constant $d(\lambda)>0$ as
\begin{align*}
d(\lambda):=\mathrm{dist}(\lambda, \sigma_{\mathrm{ess}}(U))=\displaystyle\inf_{\mu\in\sigma_{\mathrm{ess}}(U)}|\lambda-\mu|.
\end{align*}
The main result of this section is as follows:
\begin{thm}\label{thm:main}
Under Assunption \ref{Ass:main}, for any $\psi\in \mathrm{Ker}(U-\lambda)\setminus\{0\}$ with $\lambda\in \sigma_{\mathrm{d}}(U)$, $e^{\delta A}\psi\in\mathcal{H}$ for any $\delta>0$ such that $2\sinh(\delta b)<d(\lambda)$. 
\end{thm}
\begin{rem}The non-negativity in the second part of Assumption \ref{Ass:main} is not essential. However, for simplicity, we assume the non-negativity of $A$ in this paper.
\end{rem}
In what follows, we always assume Assumption \ref{Ass:main}. To prove Theorem \ref{thm:main}, we prepare some lemmas.
\begin{lem}\label{Lem:spec}
We take $\lambda\in\sigma_{\mathrm{d}}(U).$ Then for any $\epsilon>0,$ there exists $R>0$ such that
\begin{align*}
\|Uf-\lambda f\|\ge \{d(\lambda)-\epsilon\}\|f\|,
\end{align*}
for all $f\in\mathcal{H}$ such that $E_{A}([0, R))f=0$.
\end{lem}
\begin{proof}
We suppose the contrary. Then there exists $\epsilon>0$ such that for any $R>0$, there exists $f_{R}\in\mathcal{H}$ such that $\|f_{R}\|=1$, $E_{A}([0, R))f_{R}=0$ and
\begin{align*}
\|Uf_{R}-\lambda f_{R}\|<d(\lambda)-\epsilon.
\end{align*}
We choose $\theta\in[0, 2\pi)$ such that $a:=\mathrm{dist}\left(\mathrm{Arc}(\lambda, \theta), \sigma_{\mathrm{ess}}(U)\right)<d(\lambda)$ and $a>d(\lambda)-\epsilon$, where 
$$\mathrm{Arc}(\lambda, \theta):=\{\lambda e^{ik}| -\theta \le k \le \theta\}.$$
We set $X:=\mathrm{Arc}(\lambda,\theta),$ and $g_{R}:=(1-E_{U}(X))f_{R}$, where $E_{U}(\cdot)$ is the spectral measure of $U$. From the spectral theorem for unitary operators, it follows that
\begin{align*}
\|Ug_{R}-\lambda g_{R}\|^2=\displaystyle\int_{S^{1}\setminus X}|\mu-\lambda|^2\mathrm{d}\|E_{U}(\mu)g_{R}\|^2> a^2 \|g_{R}\|^2,
\end{align*}
where $S^{1}$ is the unit circle on $\C.$
Since $f_{R}$ weakly converges to 0 (as $R\to\infty$) and  $E_{U}(X)$ is compact, $E_{U}(X)f_{R}$ strongly converges to 0 (as $R\to\infty$). This implies that $\|g_{R}-f_{R}\|\to 0$ (as $R\to \infty$).
On the other hand, we have
\begin{align*}
a\|g_{R}\|&<\|Ug_{R}-\lambda g_{R}\|
\\
&\le \|Uf_{R}-\lambda f_{R}\|+\|(U-\lambda)E_{U}(X)f_{R}\|
\\
&<d(\lambda)-\epsilon +2\|E_{U}(X)f_{R}\|.
\end{align*}
By taking the limit $R\rightarrow\infty$, we get $a\le d(\lambda)-\epsilon$ since $\|g_{R}\|\rightarrow \|f_{R}\|=1\ (\mathrm{as}  \ R\rightarrow\infty).$ This is a contradiction since we took $a$ like as $a>d(\lambda)-\epsilon$.
\end{proof}
Before going to next lemma, we introduce followig step-like functions. For $N\in \mathbb{N}$ and $\delta>0,$ we define
\begin{align}\label{Def:lambda}
\Lambda(r):=\displaystyle\sum_{n=1}^{\infty}\delta  nb\I_{B_{n}}(r),\quad 
\Lambda_{N}(r):=
\begin{cases}
\displaystyle\sum_{n=1}^{N} \delta nb\mathbb{I}_{B_{n}}(r),\ &r\in[0, Nb),
\\
\delta Nb, & r\in[Nb,\infty),
\end{cases}
\end{align}
where $B_{n}:=[(n-1)b, nb)\subset \mathbb{R}$ and $\I_{B_{n}}$ is the characteristic function of $B_{n}.$ Then, $\Lambda$ approximates a function $f(r):=\delta r,\ (r\in[0, \infty))$ from the above and $\Lambda_{N}$ is a cut-off function of $\Lambda.$

For a two bounded operators $S$ and $T$, we define the commutator $[S, T]$ as $[S, T]:=ST-TS.$
\begin{lem}\label{Lem:cutoff1}
For any $R>0$, we set $E_{A}(R):=E_{A}([R, \infty)).$ Then, $e^{\Lambda(A)}[U, E_{A}(R)]$ is bounded on $\mathcal{H}$ and
\begin{align*}
\|e^{\Lambda(A)}[U, E_{A}(R)]\|\le e^{\delta \lceil R+b\rceil_{b}}+e^{\delta \lceil R\rceil_{b}},
\end{align*}
where for $x>0, \lceil x\rceil_{b}:=b\cdot \min\{n\in\mathbb{N}|\ x\le nb\}.$ 
\end{lem}
\begin{proof}
Since $U$ finitely propagates with respect to $A$, it follows that
\begin{align*}
&[U, E_{A}(R)]
\\
&=\{UE_{A}(R)-E_{A}(R)U\}
\\
&\quad \times \{E_{A}([0, R-b))+E_{A}([R-b, R))+E_{A}([R, R+b))+E_{A}(R+b)\}
\\
&=-E_{A}([R, R+b))UE_{A}([R-b, R))+E_{A}([R-b, R))UE_{A}([R, R+b)),
\end{align*}
where if $R-b\le 0,$ we set $E_{A}([0, R-b))=0$ and $E_{A}([R-b, R))=E_{A}([0, R)).$
Thus, for any $\psi \in\mathcal{H},$ it follows that $[U, E_{A}(R)]\psi\in D(e^{\Lambda(A)})$ and
\begin{align*}
\|e^{\Lambda(A)}[U, E_{A}(R)]\psi\|\le\left(e^{\delta \lceil R+b\rceil_{b}}+e^{\delta \lceil R\rceil_{b}}\right)\|\psi\|.
\end{align*}
Therefore the lemma follows.
\end{proof}
\begin{lem}\label{Lem:cutoff2}
For any $N\in\N,$ it follows that
\begin{align*}
\|[U, e^{\Lambda_{N}(A)}]e^{-\Lambda_{N}(A)}\|\le 2\sinh(\delta b).
\end{align*}
In particular, the above estimate in the right hand side does not depend on $N$.
\end{lem}
\begin{proof}
By applying the Duhamel formula,  $[U, e^{\Lambda_{N}(A)}]e^{-\Lambda_{N}(A)}$ can be expressed as
\begin{align}\label{eq:Duhamel}
[U, e^{\Lambda_{N}(A)}]e^{-\Lambda_{N}(A)}=\displaystyle\int_{0}^{1}e^{t\Lambda_{N}(A)}[U, \Lam]e^{-t\Lam}\mathrm{d}t.
\end{align}
The integrand in \eqref{eq:Duhamel} is decomposed as follows:
\begin{align*}
&e^{t\Lambda_{N}(A)}[U, \Lam]e^{-t\Lam}
\\
&= e^{t\Lam}\{U\Lam-\Lam U\}E_{A}(B_{1})
\\
&+\displaystyle\sum_{m=2}^{N}e^{t\Lam}\{ U\Lam-\Lam U\}e^{-t\Lam}E_{A}(B_{m})
\\
&+e^{t\Lam}\{U\Lam-\Lam U\}e^{-t\Lam}E_{A}(B_{N+1})
\\
&+e^{t\Lam}\{U\Lam-\Lam U\}e^{-t\Lam}E_{A}((N+1)b)
\\
&=:\mathrm{I}+\mathrm{II}+\mathrm{III}+\mathrm{IV}.
\end{align*}
The first term $\mathrm{I}$ can be calculated as follows:
\begin{align*}
\mathrm{I}&=\{E_{A}(B_{1})+E_{A}(B_{2})\}e^{t\Lam}\{U \Lam-\Lam U\}E_{A}(B_{1})
\\
&=E_{A}(B_{2})e^{2t\delta b}(U\delta b -2\delta bU)e^{-t\delta b}E_{A}(B_{1})
\\
&=-\delta b e^{t\delta b}E_{A}(B_{2})UE_{A}(B_{1}).
\end{align*}
The second term $\mathrm{II}$ can be calculated as follows:
\begin{align*}
\mathrm{II}&=\displaystyle\sum_{m=2}^{N}\{E_{A}(B_{m-1})+E_{A}(B_{m})+E_{A}(B_{m+1})\}
\\
&\quad \times e^{t\Lam}\{U\Lam-\Lam U\}e^{-\Lam}E_{A}(B_{m})
\\
&=\displaystyle\sum_{m=2}^{N}\Big[E_{A}(B_{m-1})e^{t\delta b(m-1)}\{U\delta bm-\delta b(m-1) U\}e^{-t\delta bm}E_{A}(B_{m}) 
\\
&\quad +E_{A}(B_{m+1})e^{t\delta b(m+1)}\{U\delta bm-\delta b(m+1)U \}e^{-t\delta bm} E_{A}(B_{m})\Big]
\\
&=\delta b\displaystyle\sum_{m=2}^{N}\Big[ e^{-t\delta b}E_{A}(B_{m-1})UE_{A}(B_{m})-e^{t\delta b}E_{A}(B_{m+1})UE_{A}(B_{m})\Big].
\end{align*}
The third term $\mathrm{III}$ can be calculated as follows:
\begin{align*}
\mathrm{III}&=\{E_{A}(B_{N})+E_{A}(B_{N+1})+E_{A}(B_{N+2})\}
\\
&\quad \times e^{t\Lam}\{U\Lam-\Lam U\}e^{-t\Lam}E_{A}(B_{N+1})
\\
&=E_{A}(B_{N})e^{t\delta bN}\{U\delta b(N+1)- \delta bN U\}e^{-t\delta b (N+1)}E_{A}(B_{N+1})
\\
&=\delta be^{-t\delta b}E_{A}(B_{N})UE_{A}(B_{N+1}).
\end{align*}
Lastly, the forth term $\mathrm{IV}$ can be calculated as follows:
\begin{align*}
\mathrm{IV}&=E_{A}(Nb)e^{t\Lam}\{U\Lam-\Lam U\}e^{-t\Lam}E_{A}((N+1)b)
\\
&=E_{A}(Nb)e^{t\Lam}(U b\delta N - b\delta N U)e^{-t\Lam}E_{A}((N+1)b)
\\
&=0.
\end{align*}
Thus, we get the following expression:
\begin{align*}
&[U, e^{\Lam}]e^{-\Lam}
\\
&=\delta b \displaystyle\int_{0}^{1}e^{-t\delta b}\mathrm{d}t\cdot \displaystyle\sum_{m=2}^{N+1}E_{A}(B_{m-1})UE_{A}(B_{m})-\delta b  \displaystyle\int_{0}^{1}e^{t\delta b}\mathrm{d}t\cdot\displaystyle\sum_{m=1}^{N}E_{A}(B_{m+1})UE_{A}(B_{m})
\\
&=(1-e^{-\delta b})\displaystyle\sum_{m=2}^{N+1}E_{A}(B_{m-1})UE_{A}(B_{m})-(e^{\delta b}-1)\displaystyle\sum_{m=1}^{N}E_{A}(B_{m+1})UE_{A}(B_{m}).
\end{align*}
For any $\psi\in\mathcal{H}$, we have
\begin{align*}
&\|[U, e^{\Lam}]e^{-\Lam}\psi\|^2
\\
&=\|(1-e^{-\delta b})\displaystyle\sum_{m=2}^{N+1}E_{A}(B_{m-1})UE_{A}(B_{m})\psi-(e^{\delta b}-1)\displaystyle\sum_{m=1}^{N}E_{A}(B_{m+1})UE_{A}(B_{m})\psi\|^2
\\
&=(1-e^{-\delta b})^2\displaystyle\sum_{m=2}^{N+1}\|E_{A}(B_{m-1})UE_{A}(B_{m})\psi\|^2+(e^{\delta b}-1)^2\displaystyle\sum_{m=1}^{N}\|E_{A}(B_{m+1})UE_{A}(B_{m})\psi\|^2
\\
&-2(1-e^{-\delta b})(e^{\delta b}-1)\displaystyle\sum_{m=2}^{N+1}\displaystyle\sum_{n=1}^{N}\mathrm{Re}\langle E_{A}(B_{m-1})UE_{A}(B_{m})\psi, E_{A}(B_{n+1})UE_{A}(B_{n})\psi \rangle
\\
&=(1-e^{-\delta b})^2\displaystyle\sum_{m=2}^{N+1}\|E_{A}(B_{m-1})UE_{A}(B_{m})\psi\|^2+(e^{\delta b}-1)^2\displaystyle\sum_{m=1}^{N}\|E_{A}(B_{m+1})UE_{A}(B_{m})\psi\|^2
\\
&-2(1-e^{-\delta b})(e^{\delta b}-1)\displaystyle\sum_{n=2}^{N}\mathrm{Re}\langle E_{A}(B_{n})UE_{A}(B_{n+1})\psi, E_{A}(B_{n})UE_{A}(B_{n-1})\psi \rangle
\\
&\le (1-e^{-\delta b})^2\displaystyle\sum_{m=2}^{N}\|E_{A}(B_{m})\psi\|^2+(e^{\delta b}-1)^2\displaystyle\sum_{m=1}^{N}\|E_{A}(B_{m})\psi\|^2
\\
&\quad +(1-e^{-\delta b})(e^{\delta b}-1)\displaystyle\sum_{n=2}^{N}\{\|E_{A}(B_{n})UE_{A}(B_{n+1})\psi\|^2+\|E_{A}(B_{n})UE_{A}(B_{n-1})\psi\|^2\}
\\
&\le (1-e^{\delta b})^2\|\psi\|^2+(e^{\delta b}-1)^{2}\|\psi\|^2+2(1-e^{-\delta b})(e^{\delta b}-1)\|\psi\|^2
\\
&=\{(1-e^{-\delta b})+(e^{\delta b}-1)\}^2\|\psi\|^2
\\
&=(e^{\delta b}-e^{-\delta b})^2\|\psi\|^2.
\end{align*}
Thus, the lemma follows.
\end{proof}
\begin{proof}[Proof of Theorem \ref{thm:main}]
We choose $\epsilon>0$ as $\epsilon:=[d(\lambda)-2\sinh(\delta b)]/2.$ Then, by Lemma \ref{Lem:spec}, there exists $R>0$ such that for any $f\in \mathcal{H}$ with $E_{A}([0, R))f=0, $ we have
\begin{align*}
\{d(\lambda)-\epsilon\}\|f\|\le \|Uf-\lambda f\|.
\end{align*}
We take $\psi\in\mathrm{Ker}(U-\lambda)\setminus\{0\}$ with $\lambda\in\sigma_{\mathrm{d}}(U)$. For $R$ and $b$, there exists $N_{0}\in\mathbb{N}$ such that $R<N_{0}b$. Then we set $f_{N}:=e^{\Lam}E_{A}(R)\psi,\ (N\ge N_{0}).$ Since $E_{A}([0, R))f_{N}=0$,  we have the following for arbitrary $N\ge N_{0}$:
\begin{align}\label{eq:bound}
 \{d(\lambda)-\epsilon\}\|f_{N}\|\le \|Uf_{N}-\lambda f_{N}\|.
\end{align}
From $U\psi=\lambda\psi,$ we get
\begin{align}\label{eq:comm}
Uf_{N}-\lambda f_{N}=[U, e^{\Lam}E_{A}(R)]\psi=[U, e^{\Lam}]E_{A}(R)\psi+ e^{\Lam}[U, E_{A}(R)]\psi.
\end{align}
From Lemma \ref{Lem:cutoff1}, we get
\begin{align*}
\|e^{\Lambda_{N}(A)}[U, E_{A}(R)]\psi\|\le\left (e^{\delta \lceil R+b\rceil_{b}}+e^{\delta  \lceil R\rceil_{b}}\right)\|\psi\|,
\end{align*}
For the first term of \eqref{eq:comm}, from Lemma \ref{Lem:cutoff2}, we get
\begin{align*}
\|[U, e^{\Lambda_{N}(A)}]E_{A}(R)\psi\|=\|[U, e^{\Lam}]e^{-\Lam}e^{\Lam}E_{A}(R)\psi\|\le 2\sinh(\delta b)\| f_{N}\|.
\end{align*}
Thus, we arrive at 
\begin{align*}
\|Uf_{N}-\lambda f_{N}\|\le \left(e^{\delta \lceil R+b\rceil_{b}}+e^{\delta  \lceil R\rceil_{b}}\right)\|\psi\|+2\sinh(\delta b)\|f_{N}\|.
\end{align*}
From the above inequality and \eqref{eq:bound}, we arrive at
\begin{align}\label{eq:estimate}
\displaystyle\frac{d(\lambda)-2\sinh(\delta b)}{2}\|f_{N}\|\le \left(e^{\delta \lceil R+b\rceil_{b}}+e^{\delta  \lceil R\rceil_{b}}\right)\|\psi\|.
\end{align}
Since $N$ is arbitrary and right hand side of \eqref{eq:estimate} is independent of $N$, we conclude that $e^{\Lambda(A)}\psi\in \mathcal{H}$ by the monotone convergence theorem. This implies $e^{\delta A}\psi\in\mathcal{H}.$
\end{proof}
\section{Application}
In this section, we apply the result to multi-dimensional quantum walks. We choose the Hilbert space $\mathcal{H}$ as 
\begin{align*}
\mathcal{H}:=\ell^{2}(\mathbb{Z}^{d}; \mathbb{C}^{2d}):=\left\{f:\Z^{d}\rightarrow \C^{2d}\Big|\ \displaystyle\sum_{x\in\Z^{d}}\|f(x)\|_{\C^{2d}}^2<\infty\right\}.
\end{align*}
In what follows, we freely use the identification $\mathcal{H}\simeq \oplus_{j=1}^{d}\ell^{2}(\mathbb{Z};\mathbb{C}^2)$. Thus
\begin{align*}
f(x)=\begin{bmatrix}f_{1}(x)\\ f_{2}(x) \\ \vdots \\ f_{d}(x)\end{bmatrix}=\begin{bmatrix} f_{11}(x) \\ f_{12}(x) \\ \vdots \\ f_{d1}(x) \\ f_{d2}(x) \end{bmatrix},\ f\in \mathcal{H}, \ x\in\Z^{d}.
\end{align*} 
Let $\{e_{j}\}_{j=1}^{d}$ be the set of standard orthogonal basis of $\mathbb{Z}^{d}$. Let $L_{j}$ $(j=1, \dots, d)$ be the shift operator on $j-$th direction defined by 
\begin{align*}
(L_{j}f)(x):=f(x+e_{j}),\ f\in\mathcal{H}, \ x\in\mathbb{Z}^{d}, \ j=1, \dots, d. 
\end{align*}
To introduce the shift operator $S$, we set
\begin{align*}
D:=\left\{(p, q)=(p_{1}, \dots, p_{d}, q_{1},\dots, q_{d})\in\R^{d}\times \C^{d}\Big|\ p_{j}^2+|q_{j}|^2=1,\ (j=1, \dots, d)\right\}.
\end{align*}
For $(p, q)\in D$, we define the shift operator $S$ by 
\begin{align*}
S:=S_{1}\oplus S_{2}\oplus \dots \oplus S_{d},\quad S_{j}:=\begin{bmatrix} p_{j} & q_{j}L_{j} \\ (q_{j}L_{j})^{\ast }& -p_{j}\end{bmatrix},\quad j=1,\dots, d.
\end{align*}
Next, we intoduce the coin operator $C$. Let $\{C(x)\}_{x\in\Z}\subset U(2d)$  be a set of $2d\times 2d$ self-adjoint and unitary matrices. We define the coin operator $C$ as a multiplication operator by $C(x):$
\begin{align*}
(Cu)(x):=C(x)u(x),\ u\in\mathcal{H}, \ x\in \Z.
\end{align*}
For the coin operator $C,$ we impose the following assumptioon:
\begin{assm}\label{Ass:coin}
\begin{enumerate}
\item For each $x\in\Z^d,$ 1 is a simple eigenvalue of $C(x)$, i.e., $\mathrm{dim ker}(C(x)-1)=1.$
\item There exists two self-adjoint and unitary matrices $C_{0}$ and $C_{1}$ such that
\begin{align*}
C(x)=\begin{cases} C_{1}, \ &x\in\Z^d\setminus\{0\}, \\ C_{0}, \ & x=0.\end{cases}
\end{align*}
\end{enumerate}
\end{assm}
By the first part of Assumption \ref{Ass:coin}, for each $x\in\Z^d,$ we can take a unit vector $\chi(x)$ as follows:
\begin{align*}
\chi(x)=\begin{bmatrix}\chi_{1}(x) \\ \vdots \\ \chi_{d}(x)\end{bmatrix}\in\mathrm{ker}(C(x)-1),\quad \chi_{j}(x)=\begin{bmatrix}\chi_{j1}(x) \\ \chi_{j2}(x)\end{bmatrix}\in \C^{2},\quad (j=1, \dots d).
\end{align*}
From the first part of Assumption \ref{Ass:coin} and the spectral decomposition of $C(x)$, we have $C(x)=2|\chi(x)\rangle\langle\chi(x)|-1.$ Moreover, the second part of Assumption \ref{Ass:coin} implies that $\chi$ has a form of
\begin{align*}
\chi(x)=\begin{cases}
\Phi=\begin{bmatrix}\Phi_{1}\\ \vdots \\ \Phi_{d}\end{bmatrix},\ \Phi_{j}=\begin{bmatrix}\Phi_{j1} \\ \Phi_{j2}\end{bmatrix}\in \C^2,\quad (j=1, \dots, d),\quad x\in\Z^{d}\setminus\{0\},
\\
\Omega=\begin{bmatrix}\Omega_{1}\\ \vdots \\ \Omega_{d}\end{bmatrix},\quad \Omega_{j}=\begin{bmatrix}\Omega_{j1} \\ \Omega_{j2}\end{bmatrix}\in \C^2,\quad (j=1, \dots, d),\quad x=0.
\end{cases}
\end{align*}
The condition $\mathrm{dimKer}(C(x)-1)$ is needed to construct a coisometry from $\ell^{2}(\Z^{d};\C^{2d})$ to $\ell^{2}(\Z^d;\C^d)$ and to apply the spectral mapping theorem \cite{SS}. 
\begin{assm}\label{Ass:FFS2}
Following conditions hold:
\begin{enumerate}
\item $\Phi_{j}\cdot (\sigma_{1}\Omega_{j}):=\Phi_{j1}\Omega_{j2}+\Phi_{j2}\Omega_{j1}\neq 0$ for all $j=1, \dots, d$,
\\
\item $\langle \Phi_{l}, \sigma_{+}\Omega_{l}\rangle_{\C^{2}}\neq 0$ for some $l=1, \dots, d,$
\end{enumerate}
where 
\begin{align*}
\sigma_{1}:=\begin{bmatrix} 0 & 1 \\ 1 & 0\end{bmatrix},\quad \sigma_{+}:=\begin{bmatrix}0 & 1 \\ 0 & 0\end{bmatrix}.
\end{align*}
\end{assm}
We introduce the following quantities:
\begin{align*}
a_{\Omega}(p):=\displaystyle\sum_{j=1}^{d}p_{j}\langle \Omega_{j}, \sigma_{3}\Omega_{j}\rangle_{\C^{2}},\quad a_{\Phi}(p):=\displaystyle\sum_{j=1}^{d}p_{j}\langle \Phi_{j}, \sigma_{3}\Phi_{j}\rangle_{\C^{2}},
\end{align*}
where,
\begin{align*}
\sigma_{3}:=\begin{bmatrix}1 & 0 \\ 0 & -1\end{bmatrix}.
\end{align*}
\begin{assm}\label{Ass:FFS3}
It follows that $a_{\Omega}(p_{0})\neq a_{\Phi}(p_{0})$ for some $p_{0}\in \{-1, 1\}^{d}.$
\end{assm}
\begin{rem}
In $d=1$, Assumption \ref{Ass:FFS2} and Assumption \ref{Ass:FFS3} are not compartible. For $d=1$, see \cite{FFS1}.
\end{rem}
To explain the theorem, for $l\in\{1, \dots, n\}$ stated in Assumption \ref{Ass:FFS2}, we set 
\begin{align*}
D_{l}:=\{(p, q)\in D|\ p_{l}q_{l}\neq 0\}.
\end{align*}
\begin{thm}\cite{FFS2}\label{Thm:FFS}
Let $d\ge 2$ and we assume Assumption \ref{Ass:coin}, \ref{Ass:FFS2} and \ref{Ass:FFS3}. Then,  there exists $\delta>0$ such that for any $(p, q)\in D_{l}$ with $\|(p, q)-(p_{0}, 0)\|_{\R^d\times \C^d},$ $\sigma_{\mathrm{d}}(U)\neq\emptyset.$ 
\end{thm}
We introduce the moduls of position operator as a non-negative self-adjoint operator $A$ which appeared in Assumption \ref{Ass:main}:
\begin{align*}
\mathrm{Dom}(|Q|)&:=\left\{u\in\mathcal{H}|\ \displaystyle\sum_{x\in\Z^d} |x|^2\|u(x)\|_{\C^{2d}}<\infty\right\},
\\
 (|Q|u)(x)&:=|x|u(x),\quad u\in \mathrm{Dom}(|Q|), \quad x\in\Z^d.
\end{align*}
Then, for any $0\le R_{1}< R_{2},$ and $u\in \mathrm{Ran}E_{|Q|}([R_{1}, R_{2})), $ we have $Uu\in\mathrm{Ran}E_{|Q|}([R_{1}-1, R_{2}+1)).$ Thus, we can choose the constant $b$ which appeared in Assumption \ref{Ass:main} as $b=1.$ By Theorem \ref{thm:main}, we get the following result:
\begin{thm}\label{Thm:exp}
For any $\lambda\in\sigma_{\mathrm{d}}(U)$ and $\psi\in\mathrm{Ker}(U-\lambda)\setminus\{0\},$ $e^{\delta|Q|}\psi\in\mathcal{H}$ for any $\delta>0$ with $2\sinh \delta<d(\lambda).$ 
\end{thm}
As a corollary of Theorem \ref{Thm:exp}, we can derive the pointwise estimate:
\begin{cor}\label{Coro:FFS}
Under the same assumption of Theorem \ref{Thm:FFS}, for any $\delta >0$ with $2\sinh \delta <d(\lambda)$, there exists $C_{\delta}>0$ such that
for any $x\in \Z^{d},$ it follows that
\begin{align*}
\|\psi(x)\|_{\C^{2d}}\le C_{\delta}e^{-\delta |x|}.
\end{align*} 
\end{cor}
\begin{proof}Since $\psi\in D(e^{\delta |Q|}),$ $\{e^{\delta |x|}\|\psi(x)\|_{\C^{2d}}\}_{x\in\mathbb{Z}^{d}}$ is bounded. We choose a constant $C_{\delta}>0$ as $C_{\delta}:=\sup_{x\in \mathbb{Z}^{d}}e^{\delta|x|}\|\psi(x)\|_{\C^{2d}}$. Then, it follows that
\begin{align*}
\|\psi(x)\|_{\C^{2d}}=e^{\delta|x|} e^{-\delta |x|} \|\psi(x)\|_{\C^{2d}}     \le  C_{\delta}e^{-\delta |x|}.
\end{align*}
\end{proof}
\section*{Acknowledgments}
The author acknowledges support by JSPS KAKENHI Grant Number 23K03224. This work was partially supported by the Research Institute for Mathematical Sciences, an International Joint Usage/Research Center located in Kyoto University. The author thanks the anonymous referee for careful reading and fruitful comments.


\end{document}